\documentclass[12pt,a4paper,reqno]{amsart}
\usepackage[utf8]{inputenc}
\usepackage[T1]{fontenc}
\usepackage{lmodern}
\usepackage{fullpage}

\newtheorem{theorem}{Theorem}

\newtheorem{proposition}{Proposition}
\newtheorem{remark}{Remark}
\newtheorem{lemma}{Lemma}
\newtheorem{corollary}{Corollary}

\DeclareMathOperator{\Hess}{Hess}
\DeclareMathOperator{\EHF}{{\mathcal E}^\text{HF}}
\DeclareMathOperator{\E}{{\mathcal E}}
\DeclareMathOperator{\F}{{\mathcal F}}
\DeclareMathOperator{\Span}{{Span}}

\newcommand{\R}{\mathbb{R}}
\newcommand{\C}{\mathbb{C}}

\DeclareMathOperator{\Tr}{Tr}
\DeclareMathOperator{\tr}{tr}
\DeclareMathOperator{\Gram}{Gram}

\newcommand{\mat}[1]{\begin{pmatrix}#1\end{pmatrix}}
\newcommand{\lela}{\left \langle}
\newcommand{\rira}{\right \rangle}
\newcommand{\lp}{\left(}
\newcommand{\rp}{\right)}
\newcommand{\norm}[1]{\left\|#1\right\|}
\newcommand{\abs}[1]{\left|#1\right|}
\usepackage{xspace}
\newcommand{\ie}{\textit{i.e.\xspace}}
\newcommand{\eps}{\varepsilon}

\usepackage{amsmath, amssymb, amsthm}
\usepackage{commath}

\usepackage{enumerate}

\date{\today}
\address{Université Paris-Dauphine, CEREMADE, Place du Maréchal Lattre
  de Tassigny, 75775 Paris Cedex 16, France.}
\email{levitt@ceremade.dauphine.fr}
\thanks{Support from the grant
  ANR-10-BLAN-0101 of the French Ministry of Research is gratefully
  acknowledged}
\subjclass[2010]{Primary: 35Q40; Secondary: 49S05, 81V55}
\keywords{multiconfiguration, relativistic quantum mechanics,
  variational methods}

\begin{document}
\title{Solutions of the multiconfiguration Dirac-Fock equations}
\begin{abstract}
  The multiconfiguration Dirac-Fock (MCDF) model uses a linear
  combination of Slater determinants to approximate the electronic
  $N$-body wave function of a relativistic molecular system, resulting
  in a coupled system of nonlinear eigenvalue equations, the MCDF
  equations. In this paper, we prove the existence of solutions of
  these equations in the weakly relativistic regime. First, using a
  new variational principle as well as results of Lewin on the
  multiconfiguration nonrelativistic model, and Esteban and Séré on
  the single-configuration relativistic model, we prove the existence
  of critical points for the associated energy functional, under the
  constraint that the occupation numbers are not too small. Then, this
  constraint can be removed in the weakly relativistic regime, and we
  obtain non-constrained critical points, i.e. solutions of the
  multiconfiguration Dirac-Fock equations.
\end{abstract}
\maketitle
\setcounter{tocdepth}{2}
\tableofcontents

\section{Introduction}
Consider an atom or molecule with $N$ electrons. Nonrelativistic
quantum mechanics dictates that, under the Born-Oppenheimer
approximation, the electronic rest energy is given by the lowest
fermionic eigenvalue of the $N$-body Hamiltonian. The complexity of
this problem grows exponentially with $N$, and approximations are used
to keep the problem tractable. Hartree-Fock theory uses the
variational ansatz that the $N$-body wavefunction is a single Slater
determinant. The optimization of the resulting energy over the
orbitals gives rise to a nonlinear eigenvalue problem, which is solved
iteratively.

It is well-known that this method overestimates the true ground state
energy by a quantity known as the correlation energy, whose size can
be significant in many cases of chemical interest
\cite{szabo1989modern}. This can be remedied by considering several
Slater determinants, a technique known as multiconfiguration
Hartree-Fock (MCHF) theory. This brings the model closer to the full
$N$-body problem, and, in the limit of an infinite number of
determinants, one recovers the true ground state energy.

Another source of errors is that the Hamiltonian used is
non-relativistic. Indeed, in large atoms, the core electrons reach
relativistic speeds (in atomic units, of the order of $Z$, compared
with the speed of light $c \approx 137$). This causes a length
contraction which affects the screening by the core electrons of the
attractive potential of the nucleus. This has important consequences
for the valence electrons and the chemistry of elements. Neglecting
these effects leads to incorrect conclusions, and for instance fails
to account for the difference in color between silver and gold
\cite{pyykko1979relativity}.

For a fully relativistic treatment of the electrons, one should use
quantum electrodynamics (QED). But this very precise theory is also
extremely complex for all but the simplest systems. Therefore,
physicists and chemists use approximate Hamiltonians to avoid working
in the full Fock space of QED. The multiconfiguration Dirac-Fock
(MCDF) model is obtained by using the Dirac operator in the
multiconfiguration Hartree-Fock model. It incorporates relativistic
effects into the multiconfiguration Hartree-Fock model, and has been
used successfully in a number of applications
\cite{dyall2007introduction, grant2007relativistic}.

Although these models, and more complicated ones, are used routinely
by physicists, many problems still remain in their mathematical
analysis. The first rigorous proof of existence of ground states of
the Hartree-Fock equations was given by Lieb and Simon
\cite{lieb1977} and later generalized to excited states by Lions
\cite{lions1987}. The multiconfiguration equations were studied by Le
Bris \cite{lebrismulticonf}, who proved existence in the particular
case of doubly excited states. Friesecke later proved the existence of
minimizers for an arbitrary number of determinants
\cite{friesecke2003multiconfiguration}, and Lewin generalized his
proof to excited states, in the spirit of the method of Lions
\cite{lewin2004}. For relativistic models, Esteban and Séré proved
existence of single-configuration solutions to the Dirac-Fock
equations \cite{esteban-sere-1999}, and studied their non-relativistic
limit \cite{esteban2001}. To our knowledge, the present work is the
first mathematical study of a relativistic multiconfiguration model.

The main mathematical difficulty of the multiconfiguration equations,
apart from the increased algebraic complexity, is that one cannot
simultaneously diagonalize the Fock operator and the matrix of
Lagrange multipliers. Lewin rewrote the Euler-Lagrange equations in a
vector formalism and used the same arguments as in the Hartree-Fock
case \cite{lieb1977, lions1987} to prove the existence of solutions.

The Dirac-Fock equations are considerably more difficult to handle
than the Hartree-Fock equations. The main difficulty is that the Dirac
operator is not bounded from below. This fact, which causes important
problems already in the linear theory, complicates the search for
solutions of the equations, because every critical point has an
infinite Morse index. One can therefore no longer minimize the energy
functional, or even use standard critical point theory. Esteban and
Séré \cite{esteban-sere-1999}, later generalized by Buffoni, Esteban
and Séré \cite{estebanserebuffoni}, used the concavity of the energy
with respect to the negative directions of the free Dirac operator to
reduce the problem to one whose critical points have a finite Morse
index.

The MCDF model combines the two mathematical problems and adds the
difficulty that, for the theory to make sense, the speed of light has
to be above a constant that depends on a lower bound on the occupation
numbers. Note that this difficulty with small occupation numbers is
also encountered in numerical computations
\cite{indelicato2007projection}, and theoretical studies of the
nonrelativistic evolution problem \cite{bardos2010setting}.

In this paper, we prove the existence of solutions, when the
speed of light is large enough (weakly relativistic regime). We now
describe our formalism.

\section{Definitions}

In atomic units, the Dirac operator is given by
\begin{align}
  D_{c} = -i c (\alpha \cdot \nabla) + c^{2}\beta.
\end{align}
In standard representation, $\alpha$
and $\beta$ are $4 \times 4$ matrices given by
\begin{align*}
  \alpha_{k} = \mat{0&\sigma_{k}\\\sigma_{k}&0}, \beta_{k} = \mat{I_{2}&0\\0&-I_{2}},
\end{align*}
where the $\sigma_{k}$ are the Pauli matrices
\begin{align*}
  \sigma_{1} = \mat{0&1\\1&0}, \sigma_{2} = \mat{0&-i\\i&0},
  \sigma_{3} = \mat{1&0\\0&-1}.
\end{align*}
The speed of light $c$ has the physical value $c \approx 137$.

The operator $D_{c}$ is self-adjoint on $L^{2}(\R^{3},\C^{4})$ with
domain $H^{1}(\R^{3},\C^{4})$ and form domain
$H^{1/2}(\R^{3},\C^{4})$. It verifies the relativistic identity
$D_{c}^{2} = c^{4} - c^{2} \Delta$. More precisely, it admits the
spectral decomposition
\begin{align}
  D_{c} = P^{+} \sqrt{c^{4} - c^{2}\Delta} \, P^{+}-
  P^{-} \sqrt{c^{4} - c^{2}\Delta} \,P^{-},
\end{align}
where the projectors $P^{\pm}$ are given in the Fourier domain by
\begin{align}
  P^{\pm}(\xi) = \frac 1 2 \lp 1_{\C^{4}} \pm
  \frac{c\alpha\cdot\xi + c^{2} \beta}{\sqrt{c^{4} + c^{2} \xi^{2}}}\rp.
\end{align}

We denote by
\begin{align}
  E = H^{1/2}(\R^{3},\C^{4})
\end{align}
the form-domain of $D_{c}$, and
\begin{align}
 E^{\pm} = P^{\pm} E  \label{def_Epm}
\end{align}
the two
positive and negative spectral subspaces.

We will use three scalar products in this paper:
\begin{align*}
  &\lela \psi, \phi\rira_{L^{2}} = \int_{\R^{3}} \psi^{*} \phi,\\
  &\lela \psi, \phi\rira_{E} = \lela \psi, \sqrt{1-\Delta}\phi\rira_{L^{2}},\\
  &\lela \psi, \phi\rira_{c} = \lela \psi, \sqrt{1-\frac {\Delta}{c^{2}}}\phi\rira_{L^{2}},
\end{align*}
with associated norms $\norm{\psi}_{L^{2}}, \norm{\psi}_{E},
\norm{\psi}_{c}$. The purpose of this last norm is to simplify several
estimates. It is related to the change of variables $d_{c}(\psi)(x)
= c^{-3/2}\psi(\frac x c)$ used in \cite{esteban2001} in the sense that
\begin{align*}
  \lela \Psi, \Phi\rira_{c} = \lela d_{c}(\Psi), d_{c}(\Phi)\rira_{E}
\end{align*}

A molecule made of $M$ nuclei with positions $z_{i}$ and charges
$Z_{i}$ creates an attractive potential
\begin{align*}
  V(x) = - \sum_{i=1}^{M}  \frac{Z_{i}}{\abs{x-z_{i}}}.
\end{align*}
More generally, we consider a charge distribution $\mu \geq 0$ with
$\mu(\R^{3}) = Z$, which creates a potential
\begin{align}
  V= -\mu\star \frac 1 {\abs x}.
\end{align}

In the sequel, we shall always assume that $N < Z+1$, which is the
only case where we can prove existence of solutions to our
equations. This assumption is made in existence proofs for the
Hartree-Fock model to ensure that an electron cannot ``escape to
infinity'', because it will then feel the effective attractive
potential $\frac{(N-1) - Z}{\abs x}$ \cite{lions1987,
  lieb1977}. Mathematically, it is used to prove that second order
information on Palais-Smale sequences implies that the Lagrange
multipliers are not in the essential spectrum.

The Hamiltonian $D_{c} + V$ has a spectral gap around zero as long as
\begin{align*}
  Z <  \frac 2 {\pi/2+2/\pi} c.
\end{align*}

This is related to the following Hardy-type inequality (see
\cite{tix1998strict,herbst1977spectral,kato1995perturbation}) :
\begin{align}
  \abs{\lela \psi, V \psi\rira} \leq \frac Z 2 (\pi/2+2/\pi) \lela \psi,
  \sqrt{1-\Delta}\psi\rira
\end{align}
for all $\psi \in E^{\pm}$, a refinement of the Kato inequality
\begin{align}
  \label{katoineq}
  \abs{\lela \psi, V \psi\rira} &\leq \frac {Z\pi} 2 \lela \psi,
  \sqrt{-\Delta}\psi\rira
\end{align}
for all $\psi \in E$, which we will use extensively in this paper. We
also recall the standard Hardy inequality:
\begin{align}
  \norm{V \phi}_{L^{2}} \leq 2Z \norm{\nabla \phi}_{L^{2}}
\end{align}
for all $\phi \in H^{1}$.

The $N$-body relativistic Hamiltonian is given by
\begin{align*}
  H^{N} = \sum_{i=1}^{N} (D_{c,x_{i}} + V(x_{i})) + \sum_{1 \leq i < j
    \leq N} \frac 1 {\abs{x_{i} - x_{j}}}.
\end{align*}

This Hamiltonian acts on $\bigwedge^{N} L_{a}^{2}(\R^{3},\C^{4})$, the
fermionic $N$-body space. Its interpretation is problematic, and no
self-adjoint formulation is known \cite{derezinski}.

For a given $K \geq N$, the multiconfiguration ansatz is
\begin{align}
  \label{multidf_ansatz}
  \psi = \sum_{1 \leq i_{1} < \dots < i_{N} \leq K} a_{i_{1},\dots,i_{N}} \; \left|\psi_{i_{1}}\, \dots\, \psi_{i_{N}}\right\rangle,
\end{align}
where
\begin{align*}
  \left|\psi_{i_{1}}\, \dots\, \psi_{i_{N}}\right\rangle(X_{1},\dots,X_{N}) =
  \frac{1}{\sqrt{N!}} \det(\psi_{i_{k}}(X_{l}))_{k,l}
\end{align*}
with $X_{l} = (x_{l},s_{l}), x_{l} \in \R^{3}, s_{l} \in \{1,2,3,4\}$
are Slater determinants, and $a \in S, \Psi \in \Sigma,$ where
\begin{align}
  S &= \{a \in \C^{K \choose N}, \norm{a}^{2} = \sum_{1 \leq i_{1} < \dots < i_{N} \leq K} \abs{a_{i_{1,\dots,i_{N}}}}^{2}= 1\},\\
  \Sigma &=  \{\Psi \in E^{K}, \Gram \Psi = 1\},\\
  &=  \{\Psi \in E^{K}, \lela \psi_{i}, \psi_{j}\rira_{L^{2}} = \delta_{ij}\}.\notag
\end{align}

Our convention here and in the rest of this paper is to use lower
case greek letters for orbitals $\psi \in E$, and upper case greek letters
for vectors of orbitals $\Psi \in E^{K}$. We extend in a straightforward
way the scalar products
$\lela\cdot,\cdot\rira_{L^{2}},\lela\cdot,\cdot\rira_{E}$ and
$\lela\cdot,\cdot\rira_{c}$ to the space $E^{K}$:
\begin{align*}
  \lela \Psi, \Phi\rira_{*} = \sum_{k=1}^{K} \lela \psi_{k}, \phi_{k}\rira_{*}.
\end{align*}

Following \cite{lewin2004}, we define
\begin{align*}
  \alpha_{i_{1}\dots i_{N}} =
  \begin{cases}
    0 & \text{if } \#(i_{1}\dots i_{N}) < N,\\
    \frac{\epsilon(\sigma)}{\sqrt{N!}}
    a_{i_{\sigma(1)},\dots,i_{\sigma(N)}}&\text{ otherwise},
  \end{cases}
\end{align*}
where, for all $i_{1},\dots,i_{N}$ with $\#(i_{1}\dots i_{N}) = N$,
$\sigma$ is the unique permutation such that $i_{\sigma(1)} < \dots <
i_{\sigma(N)}$.

With this definition,
\begin{align*}
  \psi(X_{1},\dots,X_{N}) = \sum_{1\leq i_{1} \leq N,\, \dots,\,1\leq i_{N} \leq N,} \alpha_{i_{1},\dots,i_{N}} \psi_{i_{1}}(X_{1})\dots\psi_{i_{N}}(X_{N}).
\end{align*}

Then, substituting into the relativistic energy $\lela \psi, H^{N}
\psi\rira$, we obtain \cite{lewin2004}
\begin{align}
  \E(a,\Psi) &= \lela \Psi, \lp D_{c}\Gamma_{a} + V \Gamma_{a} +
  W_{a,\Psi}\rp \Psi\rira_{(L^{2}(\R^{3},\C^{4}))^{K}},
\end{align}
with the $K \times K$ Hermitian matrices
\begin{align*}
  (\Gamma_{a})_{i,j} &= N \sum_{k_{2}\dots k_{N}} \alpha_{i,k_{2}\dots k_{N}}^{*} \alpha_{j,k_{2}\dots k_{N}},\\
  (W_{a,\Psi})_{i,j} &= \frac{N(N-1)}2 \sum_{k_{3}\dots k_{N}}
  \sum_{k,l} \alpha^{*}_{i,k,k_{3}\dots k_{N}} \alpha_{j,l,k_{3}\dots
    k_{N}}\lp \psi^{*}_{k}\psi_{l} \star \frac 1 {\abs x}\rp.
\end{align*}

The eigenvalues $\gamma_{i}$ of $\Gamma_{a}$, for $a \in S$, satisfy
$0 \leq \gamma_{i} \leq 1$, $\sum_{i=1}^{K} \gamma_{i} = N$. They are
called occupation numbers, and measure the total weight of the
corresponding orbital in the $N$-body wave function.

In the rest of this paper, we will write matrix inequalities in the
sense of Hermitian matrices. We also extend the notation $\liminf$ and
$\limsup$ to matrix inequalities, in an abuse of notation. For
instance, we will take $\liminf \Gamma_{n} \geq \Gamma$ to mean that
the smallest eigenvalue $\lambda_n^{1}$ of $\Gamma_{n} - \Gamma$
satisfies $\liminf \lambda_{n}^{1} \geq 0$.

For reference, we define the multiconfiguration Hartree-Fock
energy
\begin{align}
  \EHF(a,\Phi) = \lela \Phi, \lp -\frac 1 2 \Delta\, \Gamma_{a} + V \Gamma_{a} +
  W_{a,\Phi}\rp \Phi\rira_{(L^{2}(\R^{3},\C^{2}))^{K}},
\end{align}
on $S \times \{\Phi \in (H^{1}(\R^{3},\C^{2}))^{K}, \Gram \Phi
= 1\}$.

One can define a group action on $S \times \Sigma$ that leaves $\E$
and $\EHF$ invariant : for any unitary matrix $U \in \mathcal U(K)$,
\begin{align}
  \label{groupaction}
  U \cdot (a,\Psi) = (a',U \Psi),
\end{align}
where $a'$ is defined via the equivalent variables $\alpha'$ :
\begin{align*}
  \alpha'_{i_{1},\dots,i_{N}} = \sum_{j_{1},\dots,j_{N}}
  (U^{*})_{i_{1},j_{1}}\dots (U^{*})_{i_{N},j_{N}} \alpha_{j_{1},\dots,j_{N}},
\end{align*}
where $U^{*}$ is the adjoint of $U$. This group action is the
multiconfiguration analogue of the well-known unitary invariance of
the Hartree-Fock equations.

The MCDF equations, obtained as the Euler-Lagrange equations of $\E$
under the constraints $a \in S$ and $\Psi \in \Sigma$, are, for $\Psi$
and $a$ respectively,
\begin{align}
  \label{DF-1}
  H_{a,\Psi} \Psi &= \Lambda \Psi,\\
  \label{DF-2}
  \mathcal H_{\Psi} a &= E a,
\end{align}
where
\begin{align}
  H_{a,\Psi} &= D_c \Gamma_{a} + V \Gamma_{a} + 2 W_{a,\Psi}
\end{align}
is the Fock operator, and
\begin{align}
  (\mathcal H_{\Psi})_{I,J} &= \lela \psi_{i_{1}}\, \dots\,
  \psi_{i_{N}}\, \big| H^{N}\big|\, \psi_{j_{1}}\, \dots\, \psi_{j_{N}}\rira
\end{align}
are the coefficients of the ${K \choose N} \times {K \choose N}$
matrix of the $N$-body Hamiltonian $H^{N}$ in the basis of the Slater
determinants. Our goal in this paper is to prove the existence of
solutions to \eqref{DF-1} and \eqref{DF-2} by finding critical points
of $\E$ on $S \times \Sigma$.

\section{Strategy of proof}
There are several major mathematical difficulties in the study of the
MCDF model that are not present in the single-configuration case. One
can use the group action \eqref{groupaction} to diagonalize $\Gamma$
or $\Lambda$, but not both at the same time. Worse, because
$W_{a,\Psi}$ does not in general commute with $\Gamma$, one can only
prove that the Fock operator $H_{a,\Psi}$ has a spectral gap around 0
for values of $c$ that depend on a lower bound on the eigenvalues of
$\Gamma$. This gap is used centrally to prove the convergence of
Palais-Smale sequences. Therefore, one needs a lower bound on
$\Gamma$.

To obtain this lower bound, we consider the (formal) nonrelativistic
limit of the multiconfiguration Dirac-Fock model, the
multiconfiguration Hartree-Fock model. Let
\begin{align}
  \label{Ik}
  I^K &= \inf\left\{\EHF(a,\Phi), a \in S,
    \Phi \in (H^{1}(\R^{3},\C^{2}))^{K}, \Gram \Phi = 1\right\}
\end{align}
be the ground-state energy of the nonrelativistic multiconfiguration
method of rank $K \geq N$. $I^{N}$ is the Hartree-Fock energy. $I^{K}$
is non-increasing, and converges to $I^{\infty}$, the Schrödinger
energy. The behavior of $I^{K}$ is not precisely known, but a result
by Friesecke \cite{friesecke2003infinitude} shows that $I^{K+2} <
I^{K}$. Therefore, $I^{K} < I^{K-1}$ at least for one every two $K$.
When this strict inequality holds, every minimizer satisfies $\Gamma >
0$. Because of the compactness of
these minimizers (proved in \cite{lewin2004}, Theorem 1), there is a
uniform bound $\gamma_{0}>0$ such that for every minimizer,
$\Gamma_{a} \geq \gamma_{0}$.

Because there is no well-defined ``ground state energy'' in the
relativistic case, we cannot use information of this type
directly. Instead, we fix $\gamma < \gamma_{0}$, and use a min-max
principle to look for solutions in the domain
\begin{align*}
  S_{\gamma} = \{a \in S, \Gamma_{a} \geq \gamma\}.
\end{align*}

By arguments inspired by
\cite{esteban-sere-1999,esteban2001,lewin2004}, we prove that the
min-max principle yields solutions of $H_{a,\Psi} \Psi = \Lambda
\Psi$, for $c$ large enough (Theorem \ref{sol_gamma}). But these are only solutions of the
equation $\mathcal H_{\Psi} a = E a$ if the constraint is not
saturated, \ie{} if $\Gamma_{a} > \gamma$.

To prove that this is the case, we take the nonrelativistic ($c \to
\infty$) limit of the critical points found in the first step. By
arguments similar to the ones in \cite{esteban2001}, we prove that
these critical points converge, up to a subsequence, to a minimizer of
the multiconfiguration Hartree-Fock functional (Theorem
\ref{thm_cv_min}). Therefore, for $c$ large, the constraint
$\Gamma_{a} \geq \gamma$ is not saturated, and we obtain solutions of
the MCDF equations (Corollary \ref{main_corollary}).

In the rest of this paper, we will always assume that $I^{K} <
I^{K-1}$, so that $\Gamma \geq \gamma_{0}$ on the nonrelativistic
minimizers. $\gamma > 0$ is a fixed constant, taken to be less than
$\gamma_{0}$. We also assume $N < Z+1$.

First, for all $\Psi \in (L^{2})^{K}$ such that $\Gram \Psi > 0$
(linearly independent components), we define the normalization
\begin{align}
  g(\Psi) = (\Gram \Psi)^{-1/2} \Psi,
\end{align}
which has the property that $g(\Psi) \in \Sigma$. This normalization
was used in \cite{esteban2001} to prove another variational principle
for the relativistic ``ground state'', which we shall not use here.

Define
\begin{align*}
  \Sigma^{+} &= \Sigma \cap (E^{+})^{K},\\
  &= \left\{\Psi \in (P^{+} E)^{K}, \Gram \Psi = 1\right\}.
\end{align*}
We will find solutions to our equations as a result of the following
variational principle:
\begin{align}
  \label{multidfvarprinc}
  I_{c,\gamma} = \inf_{a \in S_{\gamma}, \Psi^{+} \in \Sigma^{+}}
  \sup_{\Psi^{-} \in (E^{-})^{K}} \E(a,g(\Psi^{+} + \Psi^{-})).
\end{align}
\section{Results}
Our first result is the well-posedness of our variational principle:
\begin{theorem}[Existence of solutions in $S_{\gamma}$]
  \label{sol_gamma}
  Let $N < Z+1$. There are constants $K_{1}, K_{2} > 0$ such that, for
  $c$ large enough, there is a triplet $a_{*} \in S_{\gamma},
  \Psi_{*}^{+} \in \Sigma^{+}, \Psi_{*}^{-} \in (E^{-})^{K}$ solution
  of the variational principle \eqref{multidfvarprinc}:
  \begin{align*}
    \E(a_{*},g(\Psi^{+}_{*}+\Psi_{*}^{-})) &= \max_{\Psi^{-} \in (E^{-})^{K}}
    \E(a_{*},g(\Psi_{*}^{+}+\Psi^{-})),\\
    &= \min_{a \in S_{\gamma}, \Psi^{+} \in \Sigma^{+}} \max_{\Psi^{-} \in (E^{-})^{K}}
    \E(a,g(\Psi^{+}+\Psi^{-})).
  \end{align*}
  
  Denoting $\Psi_{*} = g(\Psi_{*}^{+}+\Psi_{*}^{-})$, $\Psi_{*}$ is a
  solution of the equation $H_{a_{*},\Psi_{*}} \Psi_{*} = \Lambda_{*}
  \Psi_{*}$ in $\Sigma$.

  The Hermitian matrix of Lagrange multipliers $\Lambda_{*}$ satisfies the
  estimates
  \begin{align}
    \label{control_Lambda}
    (c^{2}-K_{1}) \Gamma_{*} &\leq \Lambda_{*} \leq (c^{2} - K_{2}) \Gamma_{*}.
  \end{align}

  Furthermore, if $\Gamma_{*} > \gamma$, then $a_{*}$ is a solution of
  $\mathcal H_{\Psi_{*}} a_{*} = I_{c,\gamma} a_{*}$.
\end{theorem}
\begin{remark}
  In theory, one could give an explicit estimate of the minimal value
  of $c$ as a function of $\gamma$. However, since there is no known
  lower bound on $\gamma_{0}$ as a function of $K$, this would be of
  little use, and considerably complexify this paper. Free from the
  need of explicit constants, we use a strategy of proof that is
  simpler than the one used in Theorem 1.2 of \cite{esteban-sere-1999}
  in the single-configuration case.
\end{remark}

We now study the nonrelativistic limit of these solutions, thanks to
the control \eqref{control_Lambda} on the Lagrange multipliers:
\begin{theorem}[Non-relativistic limit]
  \label{thm_cv_min}
  Let $I^{K} < I^{K-1}$, $N< Z+1$, $c_{n} \to \infty$, and let $(a_{n},\Psi_{n})$ be the solution
  of \eqref{multidfvarprinc} obtained by Theorem~ \ref{sol_gamma} with $c =
  c_{n}$. Then, up to a subsequence,
  \begin{align*}
    a_{n}&\to a,\\
    \Psi_{n}&\to\mat{\Phi\\0}
  \end{align*}
  in $H^{1}$ norm, where $(a,\Phi) \in S_{\gamma}\times
  (H^{1}(\R^{3},\C^{2}))^{K}$ is a minimizer of
  \begin{align}
    \label{varprincHF}
    I^K &= \inf\left\{\EHF(a,\Phi), a \in S, \Phi \in
      (H^{1}(\R^{3},\C^{2}))^{K}, \Gram \Phi = 1\right\}.
  \end{align}

  The min-max level $I_{c,\gamma}$ satisfies the asymptotics
  \begin{align*}
    I_{c,\gamma} = Nc^{2} + I^{K} + o_{c\to\infty}(1).
  \end{align*}
\end{theorem}

Since, if $I^{K} < I^{K-1}$, any minimizer of \eqref{varprincHF} must
satisfy $\Gamma \geq \gamma_{0} > \gamma$, we immediately obtain
\begin{corollary}
  \label{main_corollary}
  If $I^{K} < I^{K-1}$, $N < Z+1$, for c large enough, there are
  solutions of the multiconfiguration Dirac-Fock equations
  \eqref{DF-1}-\eqref{DF-2}.
\end{corollary}

The remainder of this paper is dedicated to the proof of Theorems
\ref{sol_gamma} and \ref{thm_cv_min}.

For Theorem \ref{sol_gamma}, we first begin with Proposition \ref{ps}, a
convergence result for Palais-Smale sequences of the functional $\E$
with Lagrange multipliers bounded away from the essential spectrum of
$D_{c} \Gamma$. Then, at $(a,\Psi^{+}) \in S_{\gamma} \times
\Sigma^{+}$ fixed, we study the variational principle
\begin{align*}
  \sup_{\Psi^{-} \in (E^{-})^{K}}\E(a,g(\Psi^{+}+\Psi^{-}))
\end{align*}
in Proposition \ref{prop_h}, under the condition that $\E(a,\Psi^{+}) \leq
Nc^{2}$. We prove in Proposition \ref{asymp_energy} an upper bound on the
asymptotic behavior of $I_{c,\gamma}$ which will enable us to restrict
to this domain, and finally, we prove in Proposition
\ref{lemma_secondorderimplybound} that Palais-Smale sequences with
Morse-type information for the functional
\begin{align*}
  \F_{a}(\Psi^{+}) = \sup_{\Psi^{-} \in
    (E^{-})^{K}}\E(a,g(\Psi^{+}+\Psi^{-}))
\end{align*}
satisfy the hypotheses of Proposition \ref{ps}, and therefore are
precompact. Their limit up to extraction is a solution of our min-max
problem \eqref{multidfvarprinc}.

To prove Theorem \ref{thm_cv_min}, we use the estimates
\eqref{control_Lambda} on the Lagrange multipliers to prove the
compactness of the sequence $(a_{n},\Psi_{n})$, and the asymptotic
behavior from Proposition \ref{lemma_secondorderimplybound} to show that the
limit is a minimizer.

\section{Proof of Theorem \ref{sol_gamma}}
Our first result is the convergence of Palais-Smale sequences with
bounds on the Lagrange multipliers. The proof proceeds as in Lemma 2.1
of \cite{esteban-sere-1999} for the single-configuration case.
\subsection{Palais-Smale sequences for the energy functional}
\begin{proposition}[Palais-Smale sequences for $\mathcal E$]
  \label{ps}
  For $c$ large enough, if $(a_{n},\Psi_{n}) \in
  S_{\gamma} \times \Sigma$ satisfies:
  \begin{enumerate}[(i)]
  \item $ H_{a_n,\Psi_{n}} \Psi_{n} - \Lambda_{n} \Psi_{n} = \Delta_{n}
    \to 0$ in $H^{-1/2}$ \label{hypoel} with $\Lambda_{n}$ Hermitian matrices,
  \item $\liminf \Lambda_n > 0$,
  \item $\limsup c^{2}\Gamma_{n} - \Lambda_n > 0$,
  \end{enumerate}
  then, up to extraction, $(a_{n},\Psi_{n}) \to (a,\Psi)$ in
  $S_{\gamma}\times \Sigma$, where $(a,\Psi)$ is a solution of
  $H_{a,\Psi} \Psi = \Lambda \Psi$.
\end{proposition}
\textit{Proof.}
\paragraph*{\textbf{Step 1 : convergence in $H^{1/2}_{\text{loc}}$}}
Let $\Psi \in E^{K}$, and $\Psi^{\pm} = P^{\pm} \Psi$. Using the
inequality \eqref{katoineq},
\begin{align*}
  \lela \Psi^{+}, H_{a_n,\Psi_{n}} \Psi^{+} \rira &\geq \lela
  \Psi^{+}, \Gamma_{n} \sqrt{c^{4}-c^{2}\Delta} \Psi^{+}\rira + \lela
  \Psi^{+}, \Gamma_{n} V \Psi^{+}\rira,\\
  &\geq \lela \Psi^{+}, \Gamma_{n} \sqrt{c^{4}-c^{2}\Delta} \Psi^{+}\rira
  -  C_{1} \norm{\Psi^{+}}_{E}^{2},\\
  &\geq (\gamma c^{2} - C_{1} c)\norm{\Psi^{+}}_{c}^{2},
\end{align*}
where $C_{1} > 0$ depends on $Z$.
Similarly, using \eqref{katoineq} again,
\begin{align*}
  \lela \Psi^{-}, H_{a_n,\Psi_{n}} \Psi^{-} \rira &\leq - (\gamma c^{2} - C_{2} c)\norm{\Psi^{-}}_{c}^{2},
\end{align*}
with $C_{2} > 0$ depending on $N$.

Now, $\Psi^{+}$ and $\Psi^{-}$ are orthogonal for the $c$ scalar
product, so
\begin{align*}
  \norm{\Psi}_{c}^{2} = \norm{\Psi^{+}}_{c}^{2} +
  \norm{\Psi^{-}}_{c}^{2} = \norm{\Psi^{+} - \Psi^{-}}_{c}^{2}.
\end{align*}

Denoting by $\norm{\cdot}_{c}^{*}$ the dual norm of
$\norm{\cdot}_{c}$,
\begin{align}
  \norm{H_{a_n,\Psi_{n}} \Psi}_{c}^{*} &\geq \text{Re}\; \frac 1 {\norm{\Psi}_{c}}\lela \Psi^{+} -
  \Psi^{-}, H_{a_n,\Psi_{n}} \Psi\rira,\notag\\
  &=  \frac 1 {\norm{\Psi}_{c}} \lp \lela \Psi^{+}, H_{a_n,\Psi_{n}}
  \Psi^{+} \rira -  \lela \Psi^{-}, H_{a_n,\Psi_{n}}
  \Psi^{-} \rira \rp,\notag\\
  &\geq \frac 1 {\norm{\Psi}_{c}} \lp c^{2} \gamma-c\max(C_{1},C_{2})\rp
  \lp\norm{\Psi^{+}}_{c}^{2} +\norm{\Psi^{-}}_{c}^{2}\rp,\notag\\
  &\geq h_{0} \norm{\Psi}_{c}\label{h0},
\end{align}
with $h_{0} > 0$ when $c$ is large enough.

We then have
\begin{align*}
  \limsup_{n\to\infty} \norm{\Psi_{n}}_{c}& \leq \limsup_{n\to\infty} \frac 1 {h_{0}} \norm{H_{a_n,\Psi_{n}}
    \Psi_{n}}_{c}^{*} ,\\
  &\leq \limsup_{n\to\infty} \frac 1 {h_{0}}  \lp\norm{\Delta_{n}}_{c}^{*} +
  \norm{\Lambda_{n} \Psi_{n}}_{L^{2}}\rp.
\end{align*}

Therefore, $\Psi_{n}$ is bounded in $c$ norm, \ie{} in
$H^{1/2}$. Extracting a subsequence, again denoted by
$(a_{n},\Psi_{n})$, we may assume that $a_{n}\to a$, $\Gamma_{n} \to
\Gamma$, $\Lambda_{n} \to \Lambda$, and $\Psi_{n} \to \Psi$ weakly in
$H^{1/2}$, strongly in $L^{p}_{\text{loc}}$, $2 \leq p < 3$.

Since $H_{a,\Psi_{n}}$ is self-adjoint from $E^{K}$ to $(E^{K})^{*}$
and bounded away from zero, it is invertible. Define $\Psi'_{n}$ by
$$H_{a,\Psi_{n}} \Psi'_{n} = \Lambda \Psi_{n}.$$ $\Psi'_{n}$ is bounded
in $H^{1/2}$, and therefore precompact in
$L^{p}_{\text{loc}}$, $2 \leq p < 3$.

We partially invert
\begin{align*}
  \Psi'_{n} = (D_{c}\Gamma + V \Gamma)^{-1}(\Lambda
  \Psi_{n} - 2W_{a,\Psi_{n}} \Psi'_{n}).
\end{align*}

From Young's inequality, $W_{a,\Psi_{n}} \Psi'_{n}$ is precompact in
$L^{p}_{\text{loc}}$, $1 \leq p < 3$, so $\Lambda \Psi_{n} -
2W_{a,\Psi_{n}} \Psi'_{n}$ is precompact in
$L^{2}_{\text{loc}}$. Therefore, $\Psi'_{n}$ is precompact in
$H^{1/2}_{\text{loc}}$. We extract again and impose $\Psi'_{n} \to
\Psi$ in $H^{1/2}_{\text{loc}}$. But since
\begin{align*}
  H_{a,\Psi_{n}} (\Psi_{n} - \Psi'_{n}) = (H_{a,\Psi_{n}} -
  H_{a_{n},\Psi_{n}}) \Psi_{n} + \Delta_{n} + (\Lambda_{n} - \Lambda) \Psi_{n}\to 0
\end{align*}
in $H^{-1/2}$, from \eqref{h0}, $\Psi_{n} \to \Psi$ in
$H^{1/2}_{\text{loc}}$.

\paragraph*{\textbf{Step 2 : convergence in $H^{1/2}$}}
We now have convergence of $\Psi_{n}$ to $\Psi$ in
$H^{1/2}_{\text{loc}}$. $\Psi$ satisfies
\begin{align*}
  H_{a,\Psi} \Psi = \Lambda \Psi.
\end{align*}

We now look at the convergence in $H^{1/2}$ by obtaining an
approximate Euler-Lagrange equation satisfied by the error $\eps_{n} = \Psi_{n}
- \Psi$. We have the Euler-Lagrange equations satisfied by $\Psi_{n}$
and $\Psi$:
\begin{align*}
  (D_{c} \Gamma + V \Gamma + 2
  W_{a,\Psi_{n}})\Psi_{n} - \Lambda \Psi_{n}
  &= \Delta'_{n},\\
  (D_{c} \Gamma + V \Gamma + 2 W_{a,\Psi})\Psi - \Lambda \Psi
  &= 0.
\end{align*}
with $\Delta'_{n} \to 0$ in $H^{-1/2}$. Subtracting and using the fact
that $\eps_{n} \to 0$ weakly in $H^{1/2}$ and strongly in
$H^{1/2}_{\text{loc}}$, we get
\begin{align}
  \label{Lntozero}
  L_{n} \eps_{n} \to 0
\end{align}
in $H^{-1/2}$, where
\begin{align}
  L_{n} = D_{c} \Gamma +  2W_{a,\Psi_{n}} - \Lambda
\end{align}
is the Hamiltonian ``at infinity'' seen by $\eps_{n}$.

We now use a concavity argument to extract information on the positive
and negative components $\eps_{n}^{\pm} = P^{\pm}\eps_{n}$ of
$\eps_{n}$ separately.

Define the quadratic functional $Q_{n}$ on $(E^{-})^{K}$ by
\begin{align*}
  Q_{n}(\delta^{-}) = \lela \eps_{n} + \delta^{-}, L_{n}
  (\eps_{n} + \delta^{-})\rira.
\end{align*}

The second order terms are
\begin{align}
  \lela \delta^{-}, L_{n} \delta^{-}\rira &= \lela \delta^{-}, (D_{c}
  \Gamma +  2W_{a,\Psi_{n}} - \Lambda) \delta^{-}\rira,\notag\\
  &\leq - (c^{2} \gamma - C_{2}c) \norm{\delta^{-}}_{c}^{2} - \lela
  \delta^{-}, \Lambda \delta^{-}\rira.
  \label{Lnconcave}
\end{align}
Since $\Lambda > 0$, we obtain that $Q_{n}$ is strictly
concave for $c$ large.

The concavity allows us to write
\begin{align*}
  \lela \eps_{n}^{+}, L_{n} \eps_{n}^{+}\rira &= Q_{n}(-\eps_{n}^{-})\\
  &\leq Q_{n}(0) - \nabla Q_{n}(0) [\eps_{n}^{-}],\\
  &= \lela \eps_{n}, L_{n} \eps_{n} \rira - 2 \lela \eps_{n}^{-},
  L_{n} \eps_{n}\rira,\\
  &\leq 3 \norm{\eps_{n}}_{E} \norm{L_{n} \eps_{n}}_{E^{*}}.
\end{align*}
Hence
\begin{align*}
  \limsup_{n\to\infty} \lela \eps_{n}^{+}, L_{n} \eps_{n}^{+}\rira \leq 0.
\end{align*}

But
\begin{align*}
  \lela \eps_{n}^{+}, L_{n} \eps_{n}^{+}\rira \geq \lela \eps_{n}^{+},
  (c^{2} \Gamma
  - \Lambda)\eps_{n}^{+}\rira
\end{align*}
Since $\Lambda < c^{2} \Gamma$, this implies convergence to 0 of
$\eps_n^{+}$ in $L^{2}$ and then in $H^{1/2}$. But, by
\eqref{Lntozero}, this implies that $L_{n} \eps_{n}^{-} \to 0$ in
$H^{-1/2}$ and therefore that $\lela \eps_{n}^{-}, L_{n}
\eps_{n}^{-}\rira \to 0$. By \eqref{Lnconcave}, we deduce
$\eps_{n}^{-} \to 0$ in $H^{1/2}$, which proves that $\Psi_{n} \to
\Psi$ strongly in $\Sigma$.  \hfill$\square$

\subsection{The reduced functional}
For $(a,\Psi^{+}) \in S_{\gamma} \times \Sigma^{+}$, define the functional
\begin{align*}
  F_{a,\Psi^{+}}(\Psi^{-}) = \E(a,g(\Psi^{+}+\Psi^{-}))
\end{align*}
on $(E^{-})^{K}$. Our goal in this section is to prove
\begin{proposition}[Maximizers of $F$]
  \label{prop_h}
  There is a constant $M_{-} > 0$ such that, for $c$ large enough, for all
  $(a,\Psi^{+}) \in S_{\gamma} \times \Sigma^{+}$ with $\E(a,\Psi^{+})
  \leq Nc^{2}$, the functional $F_{a,\Psi^{+}}$ has a unique maximizer
  $h(a,\Psi^{+})$ in $(E^{-})^{K}$. The map $h$ is smooth, and
  satisfies
  \begin{align}
    \norm{h(a,\Psi^{+})}_{c} &\leq \frac {M_{-}} c.
  \end{align}
\end{proposition}

We first begin with estimates on $\Psi^{+}$, for which we use the
property $\E(a,\Psi^{+}) \leq Nc^{2}$.

\begin{lemma}[A priori bounds on $\Psi^{+}$]
  \label{lemma_apriori_psiplus}
  There are $M_{+}, M_{D} > 0$ such that, for $c$ large enough, if
  $(a,\Psi^{+}) \in S_{\gamma} \times \Sigma^{+}$ verifies
  $\E(a,\Psi^{+}) \leq Nc^{2}$, then
  \begin{align}
    \label{controlpsiplus1}
    \norm{\Psi^{+}}_{E} &\leq M_+,\\
    \label{controlpsiplus2}
    \left. D_{c}\right|_{\Span(\{\psi^{+}_{i}\})} &\leq c^{2} + M_{D}.
  \end{align}
  \begin{proof}
    From Kato's inequality \eqref{katoineq},
    \begin{align}
      \E(a,\Psi^+) &\geq \lela \Psi^{+}, D_{c} \Gamma \Psi^{+}\rira - C\lela
      \Psi^{+}, \sqrt{-\Delta} \Psi^{+}\rira.
      \label{control_energy}
    \end{align}
    Here and in the rest of this paper, $C$ denotes various positive
    constants independent of $c$. Since $\E(a, \Psi^{+}) \leq Nc^{2}$ and
    $\lela \Psi^{+}, \Gamma \Psi^{+}\rira = N$,
    \begin{align*}
      \lela \Psi^{+}, \left(\sqrt{c^{4}-c^{2}\Delta} - c^{2} - \frac
        {C}\gamma \sqrt{-\Delta}\right)\Gamma\Psi^{+}\rira \leq 0.
    \end{align*}

    In the Fourier domain, we can write for all $0 < \alpha < c^{2}$ by the
    Cauchy-Schwarz inequality
    \begin{align*}
      \sqrt{c^{4} + c^{2}\abs{\xi}^{2}} &\geq c^{2}\left(1-\frac \alpha {c^{2}}\right) +
      c \abs \xi \sqrt{1 - \lp 1-\frac \alpha {c^{2}}\rp^{2}},\\
      &= c^{2} - \alpha + \abs \xi \sqrt{ {2\alpha} - \frac{\alpha^{2}}{c^{2}}}.
    \end{align*}

    Therefore, we obtain
    \begin{align*}
      \lela \Psi^{+}, \left(-\alpha +\left(\sqrt{\alpha -
            \frac{\alpha^{2}}{c^{2}}}  - \frac {C}\gamma\right)\sqrt{-\Delta}\right) \Gamma\Psi^{+}\rira &\leq 0,
    \end{align*}
    so
    \begin{align*}
      \lela \Psi^{+}, \sqrt{-\Delta} \Gamma\Psi^{+}\rira &\leq \frac{N\alpha}{\sqrt{\alpha -
          \frac{\alpha^{2}}{c^{2}}}  - \frac {C}\gamma}.
    \end{align*}

    Taking $\alpha > \sqrt{C/\gamma}$ and $c$ large, $\lela \Psi^{+},
    \sqrt{-\Delta} \Gamma\Psi^{+}\rira$ is bounded independently of
    $c$. Since $\Gamma \geq \gamma$, so is $\norm{\Psi^{+}}_{E}$, and
    \eqref{controlpsiplus1} is proved.

    

    Now, using \eqref{control_energy} again along with our new
    estimate \eqref{controlpsiplus1}, we have
    \begin{align*}
      \lela \Psi^{+}, D_{c} \Gamma \Psi^{+}\rira &\leq Nc^{2} +
      C M_+^{2},\\
      \lela \Psi^{+}, (D_{c}-c^{2}) \Gamma
      \Psi^{+}\rira &\leq C M_+^{2}.
    \end{align*}
    Because $\Gamma$ annd $D_{c} - c^{2}$, taken in the sense of
    operators on $E^{k}$, commute, we have
    \begin{align*}
      \lela \Psi^{+}, (D_{c}-c^{2}) 
      \Psi^{+}\rira & \leq \frac{CM_{+}^{2}}\gamma\\
      \tr A &\leq \frac{CM_{+}^{2}}\gamma
    \end{align*}
    where $A$ is the $K \times K$ Hermitian matrix
    \begin{align*}
      A_{ij} = \lela \psi_{i}^{+}, (D_{c}-c^{2})\psi_{j}^{+}\rira.
    \end{align*}
    
    $A$ is positive semi-definite and its trace is bounded by $\frac
    {C M_{+}^{2}}\gamma$, so $A \leq \frac{CM_{+}^{2}}\gamma$, hence
    the result.
\end{proof}
\end{lemma}

We now restrict our search for a maximizer to a neighborhood of zero.
\begin{lemma}[A priori bounds on $\Psi^{-}$]
  \label{lemma_apriori_psiminus}
  There is a constant $M_{-}>0$ such that, for $c$ large enough, for
  all $(a,\Psi^{+}) \in S_{\gamma} \times \Sigma^{+}$ with
  $\E(a,\Psi^{+}) \leq Nc^{2}$,
  \begin{align*}
    \sup_{\Psi^{-} \in (E^{-})^{K}} F_{a,\Psi^{+}}(\Psi^{-})
  \end{align*}
  cannot be achieved outside a ball centered on zero of size
  $\frac{M_{-}}{c}$ in the $c$ norm.
  \begin{proof}
    Let $\Psi^{-} \in (E^{-})^{K}$, $G = \Gram(\Psi^{+} + \Psi^{-})$,
    $\Psi = G^{-1/2}(\Psi^{+} + \Psi^{-})$. Using
    \eqref{controlpsiplus2},
    \begin{align*}
      \lela \Psi, D_{c} \Gamma \Psi\rira&=(c^{2} + M_D) \lela
      \Psi, \Gamma \Psi\rira + \lela \Psi, (D_{c}-c^{2}-M_D)
      \Gamma \Psi\rira,\\
      &\leq N (c^{2} + M_D) + \lela G^{-1/2} \Psi^{-},
      (D_{c}-c^{2}-M_D) \Gamma G^{-1/2}\Psi^{-}\rira,\\
      &\leq N (c^{2} + M_D) - {2\gamma c^{2}} \norm{G^{-1/2}\Psi^{-}}_{c}^{2}.
    \end{align*}

    On the other hand,
    \begin{align*}
      \lela \Psi, (V\Gamma + 2W_{a,\Psi}) \Psi\rira &\leq
      C\norm{G^{-1/2}\Psi^{+}}_{E}^{2} + C \norm{G^{-1/2}\Psi^{-}}_{E}^{2},\\
      &\leq C M_+ + C c\norm{G^{-1/2}\Psi^{-}}_{c}^{2}.
    \end{align*}

    All together,
    \begin{align*}
      F_{a,\Psi^{+}}(\Psi^{-}) &\leq Nc^{2}+NM_D + CM_+ - \lp
      {2\gamma c^{2}}-C c\rp
      \norm{G^{-1/2}\Psi^{-}}_{c}^{2}.
    \end{align*}

    But we also have
    \begin{align*}
      F_{a,\Psi^+}(0)&= \E(a,\Psi^{+})\\
      &\geq \lela \Psi^{+}, (D_{c}\Gamma + V
      \Gamma)\Psi^{+}\rira,\\
      &\geq N c^{2} - C\norm{\Psi^{+}}_{E}^{2},\\
      &\geq N c^{2} - CM_+^{2}.
    \end{align*}


    Therefore,
    \begin{align*}
      F_{a,\Psi^+}(\Psi^{-}) \leq F_{a,\Psi^+}(0) + NM_D + 2CM_+^{2} - \lp
      {2\gamma c^{2}}-C c\rp
      \norm{G^{-1/2}\Psi^{-}}_{c}^{2}.
    \end{align*}
    So, in order to have $F_{a,\Psi^+}(\Psi^{-}) \leq
    F_{a,\Psi^+}(0)$, $\Psi^{-}$ must satisfy
    \begin{align*}
      \norm{G^{-1/2}\Psi^{-}}_{c}^{2} &\leq \frac C{c^{2}},\\
      \norm{\Psi^{-}}_{c}^{2} &\leq \frac C{c^{2}}(1+\norm{\Psi_{-}}_{L^{2}}^{2}),\\
      &\leq \frac C{c^{2}}(1+\norm{\Psi_{-}}_{c}^{2}),
    \end{align*}
    and therefore
    \begin{align*}
      \norm{\Psi^{-}}_{c}^{2} = \frac C{c^{2}}.
    \end{align*}
  \end{proof}
\end{lemma}

Restricting now to this domain, we prove that $F_{a,\Psi^{+}}$ is
strictly concave:
\begin{lemma}[Concavity of $F$]
  \label{lemma_concavity}
  For $c$ large, for all $(a,\Psi^{+})\in S_{\gamma} \times \Sigma$
  such that $\E(a,\Psi^{+}) \leq Nc^{2}$, for all $\Psi^{-}$ in the
  region $\norm{\Psi^{-}}_{c} \leq \frac {M_{-}}c$, for all $\Phi^- \in
  (E^{-})^{K}$,
  \begin{align*}
    F_{a,\Psi^{+}}''(\Psi^{-})[\Phi^-,\Phi^-] \leq - \frac{\gamma c^{2}}{2}\norm{\Phi^-}_{c}^{2}.
  \end{align*}
\end{lemma}
\begin{proof}
  We have $g(\Psi^{+} + \Psi^{-}) = (1 + \frac 1 { c} B(\Psi^{-}))
  (\Psi^{+} + \Psi^{-})$, where
  \begin{align}
    B(\Psi^{-}) = c\lp \lp 1 + \Gram \Psi^{-}\rp^{-1/2} - 1\rp.
  \end{align}

  The function
  \begin{align*}
    f(x) = c\lp \frac 1 {\sqrt{1+x^{2}}} - 1 \rp
  \end{align*}
  and its derivatives are bounded independently of $c$ in the region
  $x \leq \frac {M^{-}} c$, and therefore so is $B$ in the region $\norm{\Psi^{-}}_{c} \leq \frac {M_{-}}c$. Let $\Psi^{-}$ be
  such that $\norm{\Psi^{-}}_{c} \leq \frac {M_{-}}c$. Then, for all $\Phi^- \in (E^{-})^{K}$,
  \begin{align*}
    \frac 1 2 F''_{a,\Psi^{+}}(\Psi^{-})[\Phi^-,\Phi^-]
    &= \partial_{\Psi}\E(a,\Psi)\left[ \frac {1} { c} B'(\Psi^{-})[\Phi^-] \Phi^- +
      \frac 1 {2{ c}} B''(\Psi^{-})[\Phi^-,\Phi^-](\Psi^{+}+\Psi^{-}) \right]\\
    &+\frac 1 2 \partial_{\Psi}^{2}\E(a,\Psi)\left[(1+\Gram(\Psi^{-}))^{-1/2}\Phi^- + \frac 1 { c} B'(\Psi^{-}) [\Phi^-] (\Psi^{+}
      + \Psi^{-})\right]^{2}\\
    &=\frac 1 2 \partial_{\Psi}^{2}\E(a,\Psi)[\Phi^-,\Phi^-] + O\left(c \norm{\Phi^-}_{c}^{2}\right).
  \end{align*}

  But we also have
  \begin{align*}
    \frac 1 2 \partial_{\Psi}^{2}\E(a,\Psi)[\Phi^-,\Phi^-] &= - \lela \Phi^-,
    \sqrt{c^{4}-c^{2}\Delta}\Gamma \Phi^-\rira + O\left(\,\norm{\Phi^-}_{E}^{2}\right),\\
    &\leq \lp-\gamma c^{2} + O(c)\rp \norm{\Phi^-}_{c}^{2}.
  \end{align*}

  and so the result follows for $c$ large.

\end{proof}

Proposition \ref{prop_h} is now proved as a direct consequence of Lemmas
\ref{lemma_apriori_psiplus}, \ref{lemma_apriori_psiminus} and
\ref{lemma_concavity}. The smoothness of $h$ comes from the fact that
$h(a, \Psi^{+})$ is the maximizer of $\E(a, g(\Psi^{+} + \Psi^{-}))$,
which is smooth with respect to $a, \Psi^{+}$ and $\Psi^{-}$, and
strictly concave with respect to $\Psi^{-}$.

\subsection{Asymptotic behavior of $I_{c,\gamma}$}
In order to restrict to the domain $\E(a,\Psi^{+}) \leq Nc^{2}$, we
prove that solutions of our min-max principle have to be in this
domain for $c$ large:
\begin{proposition}
  \label{asymp_energy}
  \begin{align*}
    I_{c,\gamma} \leq Nc^{2} + I^{K} + o_{c\to\infty}(1).
  \end{align*}
  In particular, for $c$ large enough,
  \begin{align}
    I_{c,\gamma} &= \inf_{a \in S_{\gamma}, \Psi^{+} \in \Sigma^{+}}
    \,\sup_{\Psi^{-} \in (E^{-})^{K}}
    \E(a,g(\Psi^{+}+\Psi^{-})),\notag\\
    &= \inf_{\substack{a \in S_{\gamma}, \Psi^{+} \in \Sigma^{+},\\\E(a,\Psi^{+})
      < Nc^{2}}} \,\sup_{\Psi^{-} \in (E^{-})^{K}}
    \E(a,g(\Psi^{+}+\Psi^{-})).\label{restrict}
  \end{align}
\end{proposition}
\begin{proof}
  Let $(a_{*},\Phi_{*}) \in S_{\gamma} \times H^{1}(\R^{3}, \C^{2})$
  be a minimizer of the nonrelativistic multiconfiguration
  Hartree-Fock functional, and $\Psi_{*} = \mat{\Phi_{*}\\0}$. Set
  \begin{align*}
    \Psi_{*}^{+} &= g(P^{+}\Psi_{*}).
  \end{align*}
  $\Psi_{*}^{+}$ belongs to $\Sigma^{+}$, and converges in $H^{1}$ to
  $\Psi_{*}$ as $c\to\infty$.

  From the concavity
  inequality $$\sqrt{1+{\abs \xi}^{2}} \leq 1 + \frac 1 2 \abs{\xi}^{2}$$ in the Fourier
  domain, we get
  \begin{align*}
    \E(a_{*},\Psi_{*}^{+}) &= \lela \Psi_{*}^{+}, (\sqrt{c^{4} -
      c^{2}\Delta} \Gamma + V
    \Gamma + W_{a,\Psi_{*}^{+}})\Psi_{*}^{+}\rira,\\
    &\leq \lela \Psi_{*}^{+}, \left(c^{2} \Gamma - \frac 1 2 \Delta
      \Gamma + V
      \Gamma + W_{a,\Psi_{*}^{+}}\right)\Psi_{*}^{+}\rira,\\
    &= Nc^{2} + \EHF(a_{*},\Psi_{*}) + o_{c\to\infty}(1),\\
    &= Nc^{2} + I^{K} + o_{c\to\infty}(1),
  \end{align*}
  where we have used $\lela \Psi_{*}^{+}, \Gamma \Psi_{*}^{+}\rira = N$
  (because $\Tr \Gamma = N$ and $\Psi_{*}^{+} \in \Sigma$), and the
  continuity of $\EHF$ in $H^{1}$ norm.

  From Proposition \ref{prop_h} and Lemma \ref{lemma_concavity}, we now have
  \begin{align*}
    I_{c,\gamma} &\leq F_{a_{*},\Psi_{*}^{+}}(h(a,\Psi_{*}^{+})),\\
    &\leq \E(a_{*},\Psi_{*}^{+}) +
    F'_{a_{*},\Psi_{*}^{+}}(0)[h(a\Psi_{*}^{+})] - \frac {\gamma c^{2}}
    2 \norm{h(a,\Psi_{*}^{+})}_{c}^{2},\\
    &\leq \E(a_{*},\Psi_{*}^{+}) + C \norm{h(a,\Psi_{*}^{+})}_{L^{2}},\\
    &\leq \E(a_{*},\Psi_{*}^{+}) + O\lp\frac 1 c\rp,
  \end{align*}
  so that
  \begin{align*}
    I_{c,\gamma}&\leq Nc^{2} + I^{K} + o_{c\to\infty}(1).
  \end{align*}

  Since $I^{K} < 0$ and, for all $a \in S_{\gamma}, \Psi^{+} \in
  \Sigma$,
  \begin{align*}
    \sup_{\Psi^{-} \in (E^{-})^{K}} \E(a,g(\Psi^{+}+\Psi^{-})) \geq \E(a,\Psi^{+}),
  \end{align*}
  \eqref{restrict} holds and the proposition is proved.
\end{proof}

\subsection{Borwein-Preiss sequences for the reduced functional}
Let
\begin{align*}
  S'_{\gamma} = \{a \in S_{\gamma}, \inf_{\Psi^{+} \in \Sigma^{+}} \E(a,\Psi^{+}) < Nc^{2}\}.
\end{align*}
For $a \in S'_{\gamma}$ fixed, we minimize the functional
\begin{align*}
  \F_{a}(\Psi^{+}) = \E(a,g(\Psi^{+}+h(a,\Psi^{+})))
\end{align*}
on the manifold
\begin{align*}
  \Sigma^{+}_{a} &= \{\Psi^{+} \in \Sigma^{+}, \E(a,\Psi^{+}) < Nc^{2}\}.
\end{align*}
For all $\Psi^{+} \in \Sigma_{a}^{+}, \Psi \in \Sigma$, define the
tangent spaces
\begin{align*}
  T_{\Psi^{+}}\Sigma_{a}^{+} &= \{\Phi^{+} \in (E^{+})^{K}, \lela
  \phi^{+}_{i}, \psi^{+}_{j} \rira= 0\text{ for all $i,j \in \{1,\dots,K\}$}\},\\
  T_{\Psi}\Sigma &= \{\Phi \in E^{K}, \lela \phi_{i}, \psi_{j}
  \rira= 0\text{ for all $i,j \in \{1,\dots,K\}$}\}.
\end{align*}

\begin{proposition}[Borwein-Preiss sequences for $\mathcal F_{a}$ are
  Palais-Smale sequences for $\mathcal E$ with control on the multipliers]
  \label{lemma_secondorderimplybound}
  There are constants $K_{1}>0, K_{2} > 0$ such that, for all $c$
  large enough, $a \in S'_{\gamma}$, if $\Psi_{n}^{+} \in
  \Sigma_a^{+}$ is a Borwein-Preiss sequence for $\F_{a}$ on
  $\Sigma_{a}^{+}$, \ie{} satisfies
  \begin{enumerate}[(i)]
  \item $\F_{a}(\Psi_{n}^{+}) \to \inf_{\Psi^{+} \in \Sigma_a^{+}}
    \F_{a}(\Psi^{+})$,
  \item\label{hypoPS}
    ${\left.\F_{a}'(\Psi^{+}_{n})\right|_{T_{\Psi^{+}_{n}}\Sigma_{a}^{+}}}
    \to 0 $ in $H^{-1/2}$,
  \item \label{hypoBP}There is a sequence $\beta_{n} \to 0$ such that the quadratic
    form $\Phi^{+} \to \Hess\F_{a}(\Psi_{n}^{+})[\Phi^{+},\Phi^{+}] +
    \beta_{n} \norm{\Phi^{+}}_{E}^{2}$ is non-negative on
    $T_{\Psi_{n}^{+}}\Sigma_{a}^{+}$,
  \end{enumerate}
  then, denoting $\Psi_{n} = g(\Psi_{n}^{+} + h(a,\Psi_{n}^{+}))$,
  \begin{enumerate}
  \item There is a sequence of Hermitian matrices $\Lambda_{n}$ such
    that $H_{a,\Psi_{n}} \Psi_{n} - \Lambda_{n} \Psi_{n} =
    \Delta_{n} \to 0$ in $H^{-1/2}$, \label{PSreduced:1}
  \item $\limsup \Lambda_{n} \leq (c^{2} - K_{2}) \Gamma_{a},$\label{PSreduced:2}
  \item $\liminf \Lambda_{n} \geq (c^{2} - K_{1}) \Gamma_{a}.$\label{PSreduced:3}
  \end{enumerate}
\end{proposition}
\begin{proof}
  Define
  \begin{align*}
    k(\Psi^{+}, \Psi^{-}) = g(\Psi^{+}+\Psi^{-})
  \end{align*}
  on $\Sigma_{a}^{+} \times (E^{-})^{K}$. From hypothesis
  \textit{(\ref{hypoPS})} and the definition of $h$, $(\Psi_{n}^{+},
  h(a,\Psi_{n}^{+}))$ is a Palais-Smale sequence for $\E(a,k(\cdot))$.

  With the same notations as before, $k(\Psi^{+}, \Psi^{-}) = (1+\frac
  1 c B(\Psi^{-}))(\Psi^{+} + \Psi^{-})$ on $\Sigma_{a}^{+} \times
  (E^{-})^{K}$, so that, for $c$ large enough,
  $k'(\Psi_{n}^{+},h(a,\Psi_{n}^{+}))[\cdot,\cdot]$ is an isomorphism
  from $T_{\Psi_{n}^{+}}\Sigma_{a}^{+} \times (E^{-})^{K}$ to
  $T_{\Psi_{n}}\Sigma$. Therefore, $\Psi_{n} = k(\Psi_{n}^{+},
  h(a,\Psi_{n}^{+}))$ is a Palais-Smale sequence for $\E$ on $\Sigma$,
  and \textit{(\ref{PSreduced:1})} is proved.
  \subsubsection*{Upper bound on the Lagrange multipliers}
  Let us now prove the upper bound \textit{(\ref{PSreduced:2})}. Our
  strategy of proof is to link the multipliers $\Lambda_{n}$ to
  second-order information on $\E$ (Step 1). Then, choosing a
  perturbation of the first orbital far from the system (Steps 2 and
  3), so that it sees an effective charge $Z - (N - 1)$, we get an
  upper bound on $\E''$, and therefore on $(\Lambda_{n})_{11}$ in
  terms of $(\Gamma_{a})_{11}$ (Step 4). Since there is nothing special
  about the choice of the first orbital, we get our upper bound on
  $\Lambda_{n}$ (Step 5).

  \paragraph*{\textbf{Step 1}}
  We first translate the conditions (\ref{hypoPS}) and (\ref{hypoBP})
  on the Borwein-Preiss sequence $\Psi_{n}^{+}$ for the functional $\mathcal F_{a}$
  into a second-order condition on the sequence $\Psi_{n}$ for
  $\mathcal E$.
  
  Let $\delta_n \in E^{+}$ be such that $\lela \delta_{n},
  (\Psi_{n}^{+})_{i}\rira = 0, \; i = 1, \dots, K$. Let $\Psi_{n}^{-} =
  h(a,\Psi_{n}^{+})$, and $G_{n} = \Gram(\Psi_{n}^{+} + \Psi_{n}^{-}) =
  1 + \Gram \Psi_{n}^{-}$. For $\eps$ small enough, define the curve
  on $\Sigma_a^{+}$
  \begin{align*}
    \Psi_{n}^{+}(\eps)= g\lp \Psi_{n}^{+} + \eps G_{n}^{1/2} (\delta,
    0, \dots, 0)\rp.
  \end{align*}

  Define the associated
  \begin{align*}
    \Psi_{n}^{-}(\eps) &= h(a,\Psi_{n}^{+}(\eps)),\\
    G_{n}(\eps) &= \Gram(\Psi_{n}^{+}(\eps) +
    \Psi_{n}^{-}(\eps)) = 1 + \Gram \Psi_{n}^{-}(\eps),\\
    \Psi_{n}(\eps) &=
    G_{n}^{-1/2}(\eps)(\Psi_{n}^{+}(\eps) + \Psi_{n}^{-}(\eps)),
  \end{align*}
  and the infinitesimal increments
  \begin{align*}
    \Phi_{n}^{+} &= \left.\frac{\dif}{\dif
        \eps}\Psi_{n}^{+}(\eps)\right|_{\eps=0},\\
    &=G_{n}^{1/2}(\delta_n,0,\dots,0),\\
    \Phi_{n}^{-} &= \left.\frac{\dif}{\dif
        \eps}\Psi_{n}^{+}(\eps)\right|_{\eps=0},\\
    &= h'(\Psi_{n}^{+})[\Phi_{n}^{+}],\\
    \Phi_{n} &= \left.\frac{\dif}{\dif
        \eps} \Psi_{n}(\eps)\right|_{\eps=0}.
  \end{align*}

  We introduce the Lagrangian
  $L_{n}(\Psi) = \E(a,\Psi) - \lela \Psi, \Lambda_{n}
  \Psi\rira$. Because $\lela \Psi_{n}(\eps), \Lambda_{n}
  \Psi(\eps)\rira = \tr \Lambda_{n}$ does not depend on $\eps$, we can compute
  \begin{align*}
    \left.\frac{\dif^{2}}{\dif
        \eps^{2}} \mathcal F_{a}(\Psi_{n}^{+}(\eps))\right|_{\eps=0} &=     \left.\frac{\dif^{2}}{\dif
        \eps^{2}} L_{n}(\Psi_{n}(\eps))\right|_{\eps=0}\\
    &= \E''(a,
    \Psi_{n})[\Phi_{n}, \Phi_{n}] - \lela \Phi_{n}, \Lambda_{n}
    \Phi_{n} \rira + L_{n}'(\Psi_{n})\left[\left.\frac{\dif^{2}}{\dif
        \eps^{2}}\Psi_{n}(\eps)\right|_{\eps=0}\right].
  \end{align*}

  From condition (\ref{hypoPS}) we know that $L_{n}'(\Psi_{n}) =
  o_{n\to\infty}(1)$ in $H^{-1/2}$ norm, and from condition
  (\ref{hypoBP}) we know that $    \left.\frac{\dif^{2}}{\dif
        \eps^{2}} \mathcal F_{a}(\Psi_{n}^{+}(\eps))\right|_{\eps=0}
    \geq o_{n\to\infty}(1)$, hence
    
  \begin{align}
    \label{res_step1}
    \E''(\Psi_{n}) [\Phi_{n},\Phi_{n}]  \geq \lela \Phi_{n}, \Lambda_{n}
    \Phi_{n}\rira + o_{n\to\infty}(1).
  \end{align}

  \paragraph*{\textbf{Step 2}}
  We seek to choose a $\delta_{n}$ that will give us an upper bound on
  $\E''(\Psi_{n})[\Phi_{n},\Phi_{n}]$.  We first establish an
  intermediate control on $\Phi_{n}^{-}$. First, note that, since $a
  \in S'_{\gamma}$ and $\Psi_{n}^{+} \in \Sigma_{a}^{+}$, the a priori
  estimates of Proposition \ref{prop_h} and Lemma
  \ref{lemma_apriori_psiplus} hold.

  Define
  \begin{align*}
    \mathcal G(\Psi^{+}, \Psi^{-}) &=F_{a,\Psi^{+}}(\Psi^{-}),\\
    &=\E(a,g(\Psi^{+}+\Psi^{-})).
  \end{align*}
  Now, for all $\Psi^{+} \in \Sigma_{a}^{+}$, $\Phi^{-} \in (E^{-})^{K}$,
  \begin{align*}
    \partial_{\Psi^{-}}\mathcal G(\Psi^{+}, h(a,\Psi^{+}))[\Phi^{-}] = 0.
  \end{align*}
  Differentiating with respect to $\Psi^{+}$, we get, for all $\Phi^{+}
  \in T_{\Psi^{+}}\Sigma$,
  \begin{align*}
    \partial^{2}_{\Psi^{-}\Psi^{+}}\mathcal G(\Psi^{+},
    h(a,\Psi^{+}))[\Phi^{-},\Phi^{+}]
    + \partial^{2}_{\Psi^{-}\Psi^{-}}\mathcal G(\Psi^{+},
    h(a,\Psi^{+}))[\Phi^{-}, h'(\Psi^{+})[\Phi^{+}]] = 0,
  \end{align*}
  and therefore, from the definition of $\mathcal G$,
    \begin{align*}
    - F''_{a,\Psi^{+}}(\Psi^{-})[\Phi^{-}, h'(\Psi^{+})[\Phi^{+}]]=
    \partial^{2}_{\Psi^{-}\Psi^{+}}\mathcal G(\Psi^{+},
    h(a,\Psi^{+}))[\Phi^{-},\Phi^{+}].
  \end{align*}

  We now apply this to $\Psi^{+} = \Psi_{n}^{+}$, $\Psi^{-} =
  \Psi_{n}^{-}$, $\Phi^{+} = \Phi_{n}^{+}$ and $\Phi^{-} =
  \Phi_{n}^{-}$, and get
  \begin{align}
    \label{eq_f_g}
    - F''_{a,\Psi_{n}^{+}}(\Psi_{n}^{-})[\Phi_{n}^{-}, \Phi_{n}^{-}] = \partial^{2}_{\Psi^{-}\Psi^{+}}\mathcal G(\Psi_{n}^{+},
    \Psi_{n}^{-})[\Phi_{n}^{-}, \Phi_{n}^{+}].
  \end{align}
  Using the classical Hardy inequality as in \cite{esteban-sere-1999}
  (4.31),
  \begin{align*}
    \partial^{2}_{\Psi^{-}\Psi^{+}}\mathcal G(\Psi_{n}^{+},
    \Psi_{n}^{-})[\Phi_{n}^{-}, \Phi_{n}^{+}] \leq O\left(\,\norm{\nabla
        \Phi_{n}^{+}}_{L^{2}}\norm{\Phi_{n}^{-}}_{L^{2}}\right),
  \end{align*}
  where the notation $O$ is for both $c$ and $n$ large.

  But, by Lemma \ref{lemma_concavity},
  \begin{align*}
    F''_{a,\Psi_{n}^{+}}(\Psi_{n}^{-})[\Phi_{n}^{-}, \Phi_{n}^{-}] \leq -\frac{\gamma c^{2}} 2 \norm{\Phi_{n}^{-}}_{c}^{2},
  \end{align*}
  from which we conclude, from \eqref{eq_f_g}, that
  \begin{align*}
    \frac{\gamma c^{2}} 2 \norm{\Phi^{-}_{n}}_{c}^{2} \leq O\left(\,\norm{\nabla
        \Phi_{n}^{+}}_{L^{2}}\norm{\Phi_{n}^{-}}_{L^{2}}\right),
  \end{align*}
  and therefore, since $\norm{\Phi_{n}^{-}}_{L^{2}} \leq \norm{\Phi_{n}^{-}}_{c}$, that
  \begin{align}
    \label{control_hprim}
    \norm{\Phi^{-}_{n}}_{c} \leq \frac 1{c^{2}} O\lp\,\norm{\nabla
      \Phi_{n}^{+}}_{L^{2}}\rp.
  \end{align}

  \paragraph*{\textbf{Step 3}}
  We now evaluate $\E''(\Psi_{n})[\Phi_{n},\Phi_{n}]$. First, we
  compute
  \begin{align}
    \Phi_{n} &= \left.\frac{\dif}{\dif
        \eps} \Psi_{n}(\eps)\right|_{\eps=0}\notag\\
    &= \left.\frac{\dif}{\dif
        \eps} G_{n}^{-1/2}(\eps)\left(\Psi_{n}^{+}(\eps) +
        \Psi_{n}^{-}(\eps)\right)\right|_{\eps=0}\notag\\
    &= G_{n}^{-1/2} (\Phi_{n}^{+} + \Phi_{n}^{-}) + \left.\frac{\dif}{\dif
        \eps} G_{n}^{-1/2}(\eps)\right|_{\eps=0}(\Psi_{n}^{+} + \Psi_{n}^{-})\notag\\
    &= G_{n}^{-1/2} (\Phi_{n}^{+} + \Phi_{n}^{-}) - \frac 1 2 \Gram
    (\Psi_{n}^{-}, \Phi_{n}^{-})(\Psi_{n}^{+} + \Psi_{n}^{-})\notag\\
    &= (\delta_n,0,\dots,0) + R_{n}^{+} + R_{n}^{-},\label{estimate_phin}
  \end{align}
where
  \begin{align*}
    R_{n}^{+} &= \left.\frac{\dif}{\dif
        \eps} G_{n}^{-1/2}(\eps)\right|_{\eps=0} \Psi_{n}^{+},\\
    R_{n}^{-}&= \left.\frac{\dif}{\dif
        \eps} G_{n}^{-1/2}(\eps)\right|_{\eps=0} \Psi_{n}^{-} +
    G_{n}^{-1/2} \Phi_{n}^{-}.
  \end{align*}
  
  Using \eqref{control_hprim}, we can estimate the remainder terms
  $R_{n}^{\pm} \in E^{\pm}$ as
  \begin{align*}
    \norm{R_{n}^{+}}_{c} &= O\lp \frac 1 {c^{3}} \norm{\nabla
      \delta_n}_{L^{2}}\rp,\\
    \norm{R_{n}^{-}}_{c} &= O\lp \frac 1 {c^{2}} \norm{\nabla \delta_n}_{L^{2}}\rp.
  \end{align*}

  Using these estimates and the same arguments as in the Step 2 of the
  proof of Theorem 1 in \cite{lewin2004}, we can now compute
  \begin{align*}
    \E''(\Psi_{n}) [\Phi_{n},\Phi_{n}] &\leq \Gamma_{11} \lp c^{2} \norm{\delta_n}_{c}^{2} +
    \lela \delta_n, (V+\rho_{n} \star \frac 1 {\abs x}) \delta_n\rira \rp\\
    &+ O\left(\frac 1 {c^{2}}
      \norm{\nabla \delta_n}_{L^{2}}^{2} + \frac 1 {c^{3}} \norm{\nabla
        \delta_n}_{L^{2}} \norm{\delta_n}_{c}\right),
  \end{align*}
  with $\int\rho_{n} = N-1$.

  Now, let $U$ be an arbitrary vector subspace of $H^{1}(\R^{3},
  \C^{4})$ consisting of functions of the form
  \begin{align*}
    \mat{f(\,\abs x)\\0\\0\\0}
  \end{align*}
  with dimension at least $K+1$. Let $U^{+}_{\lambda}$ be the positive
  projection of the dilation of $U$ of a factor $\lambda$, \ie{}
  \begin{align*}
    U_{\lambda} = P^{+}\left\{\psi\lp\frac x \lambda\rp, \psi \in U\right\}.
  \end{align*}

  $U^{+}_{\lambda}$ is also of dimension $K+1$ for $c$ large enough,
  so we can find a function $\delta_n \in U^{+}_{\lambda}$ normalized
  in $c$ norm which is orthogonal to $\Psi_{n}^{+}$ in $L^{2}$. For
  such a function, following the estimates in Lemma 4.5 of
  \cite{esteban-sere-1999},
  \begin{align*}
    \E''(\Psi_{n}) [\Phi_{n},\Phi_{n}] &\leq \Gamma_{11} \lp c^{2} - \eta \frac {Z - (N -
      1)}{\lambda} + O\lp\frac {1} {\lambda^{2} c^{2}} + \frac {1}{\lambda c^{3}}\rp\rp,
  \end{align*}
  with $\eta > 0$, where the $O$ notation is understood for $n, c$ and
  $\lambda$ large. So, taking $\lambda$ large enough independently of
  $n$ and $c$, we get
  \begin{align}
    \label{upperbound_Eprimprim}
    \E''(\Psi_{n}) [\Phi_{n},\Phi_{n}] &\leq \lp c^{2} - \kappa\rp\Gamma_{11},
  \end{align}
  with $\kappa > 0$ independent of $n$ and $c$.

  \paragraph*{\textbf{Step 4}} We now evaluate $\lela \Phi_{n},
  \Lambda_{n} \Phi_{n}\rira$.  Using the expression
  \eqref{estimate_phin} of $\Phi_{n}$ again, we estimate
  \begin{align}
    \label{phinlambdanphin}
    \lela \Phi_{n}, \Lambda_{n} \Phi_{n}\rira &= (\Lambda_{n})_{11} + O\lp\frac 1 {c^{4}} \Lambda_{n}\rp.
  \end{align}
  But we can obtain a very crude control on $\Lambda_{n}$ thanks to the
  estimates in Lemmas \ref{lemma_apriori_psiplus} and
  \ref{lemma_apriori_psiminus}:
  \begin{align*}
    (\Lambda_{n})_{ij} &= \lela \Psi_{n}^{i}, H_{a,\Psi_{n}}
    \Psi_{n}^{j}\rira + o_{n\to\infty}(1),\\
    &= c^{2} \Gamma_{ij} + O(\,\norm{\Psi_{n}}_{E}^{2}) +
    o_{n\to\infty}(1),
  \end{align*}
  and therefore
  \begin{align*}
    \Lambda_{n} &= c^{2}\Gamma + O(1),
  \end{align*}
  so that, with \eqref{phinlambdanphin},
  \begin{align}
    \label{estimate_lambdan}
    \lela \Phi_{n}, \Lambda_{n} \Phi_{n}\rira &= (\Lambda_{n})_{11} + o_{n\to\infty}(1).
  \end{align}

  \paragraph*{\textbf{Step 5}}
  Plugging \eqref{upperbound_Eprimprim} and \eqref{estimate_lambdan}
  into \eqref{res_step1}, we obtain for $c$ and $n$ large enough
  \begin{align*}
    (\Lambda_{n})_{11} \leq (c^{2} - K_{2})\Gamma_{11},
  \end{align*}
  with $K_{2} > 0$.

  Using the group action \eqref{groupaction}, we could apply the same
  procedure to $(\widetilde a, \widetilde{\Psi}_{n}^{+}) = U \cdot
  (a,\Psi_{n}^{+})$ for any $U \in \mathcal U(K)$, and obtain
  \begin{align*}
    (U\Lambda_{n}U^{*})_{11} \leq (c^{2} - K_{2})(U\Gamma U^{*})_{11},
  \end{align*}
  which proves our result
  \begin{align*}
    \Lambda_{n} \leq (c^{2} - K_{2})\Gamma.
  \end{align*}

  \subsubsection*{Lower bound on the Lagrange multipliers}

  Let $A_{n} = (c^{2}-K_{2})\Gamma - \Lambda_{n}$. We know that, for
  $n$ large enough, $A_{n} \geq 0$, and, from
  \eqref{estimate_lambdan},
  \begin{align*}
    \tr A_{n} &= Nc^{2}-NK_{2}-\tr \Lambda_{n},\\
    &= O(1).
  \end{align*}
  So $A_{n} = O(1)$, and therefore $\Lambda_{n} \geq
  (c^{2}-K_{2})\Gamma - O(1)$. Because $\Gamma \geq \gamma > 0$, the
  result follows for $c$ large.
\end{proof}
\subsection{Proof of Theorem \ref{sol_gamma}}
For any $a \in S'_{\gamma}$, we can apply the Borwein-Preiss
variational principle \cite{borwein1987smooth} to the functional
$\F_{a}$ on $\Sigma_a^{+}$, and obtain a sequence $\Psi_{n}^{+}$ that
satisfies the hypotheses of Proposition
\ref{lemma_secondorderimplybound}. The associated sequence $\Psi_{n}$
satisfies the hypotheses of Proposition \ref{ps} so, the sequence
$(a,\Psi_{n}^{+})$ converges up to extraction to a limit
$\Psi_{a}^{+}$, solution of the min-max principle
\begin{align*}
  \F_{a}(\Psi_{a}^{+}) &= \min_{\Psi^{+} \in \Sigma^{+}} \max_{\Psi^{-} \in (E^{-})^{K}}
  \E(a,g(\Psi^{+}+\Psi^{-})).
\end{align*}
We now take a minimizing sequence $a_{n}$ for the continuous
functional $F_{a}(\Psi_{a}^{+})$ on $S'_{\gamma}$. The sequence
$(a_{n},\Psi_{a_{n}}^{+})$ again verifies the hypotheses of Proposition
\ref{ps}, and therefore converges to $(a_{*},\Psi^{+}_{*})$. The
triplet $(a_{*},\Psi^{+}_{*}, h(a_{*},\Psi_{*}^{+}))$ is now a solution of
the variational principle \eqref{multidfvarprinc}, and Theorem
\ref{sol_gamma} is proved.
\section{Proof of Theorem \ref{thm_cv_min}}
\label{sec:nonrel_limit}
\subsection{Nonrelativistic limit}
We begin with a proposition that is the multiconfiguration analogue of
Theorem 3 of \cite{esteban2001}.
\begin{proposition}[Nonrelativistic limit of solutions]
  \label{lem_cv_sols}
  Let $c_{n} \to \infty, (a_{n},\Psi_{n}) \in S_{\gamma}
  \times \Sigma$ solutions of
  \begin{align}
    \label{sol_df_n}
    H_{a_{n},\Psi_{n}} \Psi_{n} = \Lambda_{n} \Psi_{n}
  \end{align}
  such that
  \begin{align*}
    (c_n^2-K_{1}) \Gamma_{n} \leq \Lambda_{n} \leq (c_n^2 - K_{2})
    \Gamma_{n}
  \end{align*}
  for constants $K_{1}, K_{2} > 0$.

  Then, up to a subsequence, $a_{n} \to a \in S_{\gamma}$, $\Psi_{n}
  \to \mat{\Phi\\0}$ in $H^{1}$, and $\E(a_{n},\Psi_{n}) - Nc^{2} \to
  \EHF(a,\Phi)$.

\end{proposition}

\begin{proof}
  First, we need a uniform bound on $\Psi_{n}$ in $H^{1}$.
  \begin{align*}
    \norm {D_{c} \Gamma_{n} \Psi_{n}}_{L^{2}}^{2} &= \lela \Gamma_{n}
    \Psi_{n}, (c^{4} - c^{2} \Delta)\Gamma_{n}\Psi_{n}\rira,\\
    &=c_n^4 \norm{\Gamma_{n}\Psi_{n}}_{L^{2}}^{2} + c_n^2
    \norm{\Gamma_{n} \nabla
      \Psi_{n}}_{L^{2}}^{2}.
  \end{align*}

  On the other hand,
  \begin{align*}
    \norm {D_{c} \Gamma_{n} \Psi_{n}}_{L^{2}}^{2} &= \norm{(V \Gamma +
      2W_{a_n,\Psi_{n}}) \Psi_{n} - \Lambda_{n}
      \Psi_{n}}_{L^{2}}^{2},\\
    &\leq c_n^4 \norm{\Gamma_{n}\Psi_{n}}_{L^{2}}^{2} + C
    \norm{\nabla \Psi_{n}}_{L^{2}}^{2} + C c_n^2 \norm{\Gamma_{n}
      \nabla \Psi_{n}}_{L^{2}}.
  \end{align*}
  by the classical Hardy inequality, with $C>0$. Therefore, $\Psi_{n}$
  is bounded in $H^{1}$.

  We now write $\Psi_{n} = \mat{\Phi_{n}\\\mathcal X_{n}}$, where
  $\Phi_{n},\mathcal X_{n} \in H^{1}(\R^{3},\C^{2})$. We rewrite the
  equations \eqref{sol_df_n} as
  \begin{align}
    \label{phix1}
    {c_n} \Gamma_n L \mathcal X_{n} + (V\Gamma_n+2W_{a_n,\Psi_{n}}) \Phi_{n} =
    (\Lambda_{n} - c_n^2 \Gamma_{n}) \Phi_{n},\\
    \label{phix2}
    {c_n} \Gamma_n L \Phi_{n} + (V\Gamma_n+2W_{a_n,\Psi_{n}}) \mathcal X_{n} =
    (\Lambda_{n} + c_n^2 \Gamma_{n}) \mathcal X_{n},
  \end{align}
  with the operator
  \begin{align}
    L = -i \nabla \cdot \sigma.
  \end{align}

  Because $\Lambda_{n} < (c_n^2-K_{2}) \Gamma_{n}$, using
  the Hardy inequality and the boundedness of $\Phi_{n}$ in $H^{1}$,
  the first equation \eqref{phix1} yields
  \begin{align}
    \norm{\Gamma_n L \mathcal X_{n}}_{L^{2}} = \norm{\Gamma_n \nabla \mathcal
      X_{n}}_{L^{2}} = O(1/{c_n}).
  \end{align}

  The second equation \eqref{phix2} gives
  \begin{align}
    \notag
    \mathcal X_{n} &= \frac 1 {2c} \left(\frac 1 2 \left(\Gamma_{n} +
        \Lambda_{n}/c_n^2\right)\right)^{-1} \Gamma_{n} L \Phi_{n} + \frac 1 {c_n^2}
    O(\,\norm{\mathcal X_{n}}_{H^{1}})\\
    \label{approxkb}
    &= \frac 1 {2c_{n}} L \Phi_{n} + \frac 1 {c_n^2}
    O(\,\norm{\mathcal X_{n}}_{H^{1}}) + O\left(\frac 1 {c_n^3}\right)
  \end{align}
  in $L^{2}$ norm.

  Equation \eqref{approxkb} gives $\norm{\mathcal X_{n}}_{L^{2}} =
  \frac 1 {2c_{n}} \norm{L \Phi_{n}}_{L^{2}} + O(1/c_n^2)= O(1/{c_n})$, and then
  \begin{align}
    \label{approxkb_cthree}
    \mathcal X_{n} = \frac 1 {2c_{n}} L \Phi_{n} + O\left(\frac 1 {c_n^3}\right)
  \end{align}
  again in $L^{2}$ norm.

  Inserting this into the first equation \eqref{phix1} and using the
  identity $L^{2} = -\Delta$, we get the equation for $\Phi_{n}$:
  \begin{align*}
    \lp-\frac 1 2 \Delta \Gamma_{n} + V\Gamma_{n} + 2W_{\Phi_{n}}\rp \Phi_{n}
    &= (\Lambda_{n} - c_n^2 \Gamma_{n}) \Phi_{n} + \Delta_{n}\\
    \Gram \Phi_{n} &= 1 + o(1)
  \end{align*}
  with $\Delta_{n} \to 0$ in $L^{2}$ and therefore $H^{-1}$
  norm. $(a_{n},g(\Phi_{n}))$ is a Palais-Smale sequence for the
  nonrelativistic functional, with control on the Lagrange multipliers
  $(\Lambda_{n} - c_n^2 \Gamma_{n}) < 0$ and non-degeneracy
  information $\Gamma_{n} \geq \gamma$. By the arguments in the proof
  of Theorem 1, step 3 of \cite{lewin2004}, $(a_{n},\Phi_{n})$
  converges, up to a subsequence, to $(a,\Phi)$ in $H^{1}$ norm, and
  it is easy to compute from \eqref{approxkb_cthree} that
  \begin{align*}
    \lela \Psi_{n}, D_{{c_n}} \Gamma_{n} \Psi_{n}\rira = N c_n^2 +
    \frac 1 2 \lela \Phi_{n}, (-\Delta) \Gamma_{n}\Phi_{n}\rira + o(1),
  \end{align*}
  and the result follows.
\end{proof}

We are now ready to prove Theorem \ref{thm_cv_min}.

\subsection{Proof of Theorem \ref{thm_cv_min}}
\begin{proof}
  The sequence $(a_{n},\Psi_{n})$ satisfies the hypotheses of Proposition
  \ref{lem_cv_sols} : up to a subsequence, it converges strongly in
  $H^{1}$ to $\lp a,\mat{\Phi\\0}\rp$, with $\lim \E(a_{n},\Psi_{n}) -
  N c_n^2 = \EHF(a,\Phi)$. But since by Proposition \ref{asymp_energy}
  we have
  \begin{align*}
    \E(a_{n},\Psi_{n}) = I_{c_{n},\gamma} \leq Nc_{n}^{2} + I^{K} +
    o_{c_{n}\to\infty}(1),
  \end{align*}
  we obtain $\EHF(a,\Phi) = I^{K}$, hence the result.
\end{proof}

\section*{Acknowledgements}
I would like to thank Éric Séré for his attention and helpful advice.

\bibliographystyle{alpha} \bibliography{refs}

\begin{thebibliography}{BCMT10}

\bibitem[BCMT10]{bardos2010setting}
C.~Bardos, I.~Catto, N.~Mauser, and S.~Trabelsi.
\newblock Setting and analysis of the multi-configuration time-dependent
  {H}artree-{F}ock equations.
\newblock {\em {A}rchive for {R}ational {M}echanics and {A}nalysis},
  198(1):273--330, 2010.

\bibitem[BES06]{estebanserebuffoni}
B.~Buffoni, M.J. Esteban, and E.~S\'er\'e.
\newblock Normalized solutions to strongly indefinite semilinear equations.
\newblock {\em Adv. Nonlinear Stud.}, 6:323--347, 2006.

\bibitem[BP87]{borwein1987smooth}
J.M. Borwein and D.~Preiss.
\newblock A smooth variational principle with applications to
  subdifferentiability and to differentiability of convex functions.
\newblock {\em Trans. Amer. Math. Soc}, 303(51):7--527, 1987.

\bibitem[Der12]{derezinski}
J.~Derezinski.
\newblock Open problems about many-body {D}irac operators.
\newblock {\em Bulletin of International Association of Mathematical Physics},
  2012.

\bibitem[DFJ07]{dyall2007introduction}
K.G. Dyall and K.~Faegri~Jr.
\newblock {\em Introduction to relativistic quantum chemistry}.
\newblock Oxford University Press, 2007.

\bibitem[ES99]{esteban-sere-1999}
M.J. Esteban and E.~S{\'e}r{\'e}.
\newblock {S}olutions of the {D}irac-{F}ock equations for atoms and molecules.
\newblock {\em Communications in Mathematical Physics}, 203(3):499--530, 1999.

\bibitem[ES01]{esteban2001}
M.J. Esteban and E.~Séré.
\newblock {N}onrelativistic limit of the {D}irac-{F}ock equations.
\newblock {\em Annales Henri Poincare}, 2:941--961, 2001.

\bibitem[Fri03a]{friesecke2003multiconfiguration}
G.~Friesecke.
\newblock The multiconfiguration equations for atoms and molecules: charge
  quantization and existence of solutions.
\newblock {\em {A}rchive for {R}ational {M}echanics and {A}nalysis},
  169(1):35--71, 2003.

\bibitem[Fri03b]{friesecke2003infinitude}
G.~Friesecke.
\newblock On the infinitude of non-zero eigenvalues of the single-electron
  density matrix for atoms and molecules.
\newblock {\em Proceedings of the Royal Society of London. Series A:
  Mathematical, Physical and Engineering Sciences}, 459(2029):47--52, 2003.

\bibitem[Gra07]{grant2007relativistic}
I.~Grant.
\newblock {\em Relativistic Quantum Theory of Atoms and Molecules}.
\newblock Springer, 2007.

\bibitem[Her77]{herbst1977spectral}
I.~Herbst.
\newblock Spectral theory of the operator $(p^2 + m^2)^{-1/2} - ze^2/r$.
\newblock {\em Communications in Mathematical Physics}, 53(3):285--294, 1977.

\bibitem[ID93]{indelicato2007projection}
P.~Indelicato and J.P. Desclaux.
\newblock Projection operator in the multiconfiguration {D}irac-{F}ock method.
\newblock {\em Physica Scripta}, 1993(T46):110, 1993.

\bibitem[Kat66]{kato1995perturbation}
T.~Kato.
\newblock {\em Perturbation theory for linear operators}.
\newblock Springer Verlag, 1966.

\bibitem[LB94]{lebrismulticonf}
C.~Le~Bris.
\newblock A general approach for multiconfiguration methods in quantum
  molecular chemistry.
\newblock {\em Annales de l'Institut Henri Poincar{\'e}. Analyse non
  lin{\'e}aire}, 11(4):441--484, 1994.

\bibitem[Lew04]{lewin2004}
M.~Lewin.
\newblock {S}olutions of the multiconfiguration equations in quantum chemistry.
\newblock {\em {A}rchive for {R}ational {M}echanics and {A}nalysis},
  171(1):83--114, 2004.

\bibitem[Lio87]{lions1987}
P.L. Lions.
\newblock {S}olutions of {H}artree-{F}ock equations for {C}oulomb systems.
\newblock {\em Communications in Mathematical Physics}, 109(1):33--97, 1987.

\bibitem[LS77]{lieb1977}
E.H. Lieb and B.~Simon.
\newblock {T}he {H}artree-{F}ock theory for {C}oulomb systems.
\newblock {\em Comm. Math. Phys.}, 53(3):185--194, 1977.

\bibitem[PD79]{pyykko1979relativity}
P.~Pyykko and J.P. Desclaux.
\newblock Relativity and the periodic system of elements.
\newblock {\em Accounts of Chemical Research}, 12(8):276--281, 1979.

\bibitem[SO89]{szabo1989modern}
A.~Szabo and N.S. Ostlund.
\newblock {\em {M}odern quantum chemistry}.
\newblock McGraw-Hill New York, 1989.

\bibitem[Tix98]{tix1998strict}
C.~Tix.
\newblock Strict positivity of a relativistic {H}amiltonian due to {B}rown and
  {R}avenhall.
\newblock {\em Bulletin of the London Mathematical Society}, 30(3):283--290,
  1998.

\end{thebibliography}
\end{document}